\documentclass[12pt]{article}
\usepackage{amsmath}
\usepackage{amssymb, amscd, amsthm,amsfonts}
\usepackage[all]{xy}
\usepackage[dvips]{graphicx}
\usepackage{verbatim}
\newtheorem{theo}[]{{\emph{Theorem}}}

\newtheorem{lemma}[]{{\emph{Lemma}}}

\theoremstyle{remark}
\newtheorem*{remark}{\textbf{Remark}}

\theoremstyle{definition}

\newcommand{\ra}{\rightarrow}

\newcommand{\calF}{\mathcal{F}} 

\newcommand{\bF}{\mathbb{F}}

\newcommand{\cC}{\mathcal{C}}

\newcommand{\al}{\alpha}
\newcommand{\be}{\beta}
\newcommand{\ga}{\gamma}
\newcommand{\ka}{\kappa}
\newcommand{\veps}{\varepsilon}

\newcommand{\om}{\omega}

\DeclareMathOperator{\Tra}{\mathrm Tr}

\begin{document}

\title{Cyclic Codes and Sequences: the Generalized Kasami Case} \maketitle

\begin{center}
$\mathrm{Jinquan\;\;Luo\qquad\quad Yuansheng \;\;Tang \quad\qquad
and\qquad Hongyu\;\; Wang}$ \footnotetext{The authors are with the
School of Mathematics, Yangzhou University, Jiangsu Province,
225009, China
\par J.Luo is also with the Division of Mathematics, School
of Physics and Mathematical Sciences, Nanyang Technological
University, Singapore.
\par \quad E-mail addresses: jqluo@ntu.edu.sg, ystang@yzu.edu.cn, hywang@yzu.edu.cn.}
\end{center}
\newpage
 \textbf{Abstract} \par Let $q=2^n$ with $n=2m$ .
Let $1\leq k\leq n-1$ and $k\neq m$. In this paper we determine
the value distribution of following exponential sums
\[\sum\limits_{x\in
\bF_q}(-1)^{\Tra_1^m (\alpha x^{2^{m}+1})+\Tra_1^n(\beta
x^{2^k+1})}\quad(\alpha\in \bF_{2^m},\beta\in \bF_{q})\]
 and
\[\sum\limits_{x\in
\bF_q}(-1)^{\Tra_1^m (\alpha x^{2^{m}+1})+\Tra_1^n(\beta
x^{2^k+1}+\ga x)}\quad(\alpha\in \bF_{2^m},\beta,\ga\in \bF_{q})\]

 where $\Tra_1^n: \bF_q\ra \bF_2$ and $\Tra_1^m: \bF_{p^m}\ra\bF_2$ are the canonical trace
mappings. As applications:
 \begin{itemize}
    \item[(1).]  We determine the weight distribution
of the binary cyclic codes $\cC_1$ and $\cC_2$ with parity-check
polynomials $h_2(x)h_3(x)$ and $h_1(x)h_2(x)h_3(x)$ respectively
where $h_1(x)$, $h_2(x)$ and $h_3(x)$ are the minimal polynomials
of $\pi^{-1}$, $\pi^{-(2^k+1)}$ and $\pi^{-(2^m+1)}$ over
$\bF_{2}$ respectively for a primitive element $\pi$ of $\bF_q$.
    \item[(2).]We determine the correlation distribution among
    a family of m-sequences. \end{itemize}
 This paper is the binary version of Luo, Tang and Wang\cite{Luo Tan} and
extends the results in Kasami\cite{Kasa1}, Van der Vlugt\cite{Vand2}
and Zeng, Liu and Hu\cite{Zen Liu}.

\newpage
\section{Introduction}

\quad Basic results on finite fields could be found in \cite{Lid
Nie}. The following standard notations are fixed throughout this
paper except for specific statements.
\begin{itemize}
  \item Let $n$ be an even integer, $m=n/2$ and $q=2^n$.
  \item Let $\bF_{2^l}$ be the finite field of order $2^l$, $\bF_{2^l}^*$ be the set consisting of nonzero element and
   $\Tra_i^j:\bF_{p^i}\ra\bF_{p^j}$ be the trace mapping for $i\mid j$.
  \item Let $k$ be a positive integer, $1\leq k\leq n-1$ and  $k\neq m$.
  Let $d=\gcd(m,k)$ and $d'=\gcd(m+k,2k)$, $q_0=2^d$, and $s=n/d$.
  \item  Let $\pi$ be a primitive element of $\bF_q^*$.
\end{itemize}

 For binary cyclic code $\cC$ with length $l$, let $A_i$ be the
number of codewords in $\cC$ with Hamming weight $i$. The weight
distribution $\{A_0,A_1,\cdots,A_l\}$ is an important research
object for both theoretical and application interests in coding
theory. Classical coding theory reveals that the weight of each
codeword can be expressed by binary exponential sums so that the
weight distribution of $\cC$ can be determined if the corresponding
exponential sums can be calculated explicitly (Kasami \cite{Kasa1},
\cite{Kasa2}, \cite{Kasa3},  van der Vlugt \cite{Vand2}, \cite{Zen
Liu}).

In general, let $q=2^n$, $\cC$ be the binary cyclic code with length
$l=q-1$ and parity-check polynomial
\[h(x)=h_1(x)\cdots h_u(x)\quad (u\geq 1)\]
where $h_i(x)$ $(1\leq i\leq u)$ are distinct irreducible
polynomials in $\bF_{2}[x]$ with the same degree $e_i$ $(1\leq i\leq
u)$, then $\mathrm{dim}_{\bF_{2}}\cC=\sum\limits_{i=1}^{u}e_i$ .Let
$\pi^{-s_i}$ be a zero of $h_i(x)$, $1\leq s_i\leq q-2$ $(1\leq
i\leq u).$ Then the codewords in $\cC$ can be expressed by
\[c(\alpha_1,\cdots,\alpha_u)=(c_0,c_1,\cdots,c_{l-1})\quad (\alpha_1,\cdots,\alpha_u\in \bF_q)\]
where
$c_i=\sum\limits_{\lambda=1}^{u}\Tra^n_{1}(\alpha_{\lambda}\pi^{is_{\lambda}})$
$(0\leq i\leq l-1)$. Therefore the Hamming weight of the codeword
$c=c(\alpha_1,\cdots,\alpha_u)$ is {\setlength\arraycolsep{2pt}
\begin{eqnarray} \label{Wei}
w_H\left(c\right)&=& \#\left\{i\left|0\leq i\leq l-1,c_i\neq
0\right.\right\}
\nonumber\\[1mm]
&=& l-\#\left\{i\left|0\leq i\leq l-1,c_i=0\right.\right\}
\nonumber\\[1mm]
&=&
l-\frac{1}{2}\,\sum\limits_{i=0}^{l-1}\sum\limits_{a=0}^{1}(-1)^{a\cdot\Tra_{1}^n\left(\sum\limits_{\lambda=1}^{u}\alpha_{\lambda}\pi^{is_{\lambda}}\right)}
\nonumber\\[1mm]
&=&l-\frac{l}{2}-\frac{1}{2}\,\sum\limits_{x\in
\bF_q^*}(-1)^{\Tra_1^n(f(x))} \nonumber
\\[1mm]
&=&\frac{l}{2}+\frac{1}{2}-\frac{1}{2}\,S(\alpha_1,\cdots,\alpha_u)\nonumber\\[1mm]
&=&2^{n-1}-\frac{1}{2}\,S(\alpha_1,\cdots,\alpha_u)
\end{eqnarray}
} where
$f(x)=\alpha_1x^{s_1}+\alpha_2x^{s_2}+\cdots+\alpha_ux^{s_u}\in
\bF_{q}[x]$, $\bF_q^*=\bF_q\backslash\{0\}$,  and
\[S(\alpha_1,\cdots,\alpha_u)=\sum\limits_{x\in \bF_q}(-1)^{\Tra_1^n(\alpha_1x^{s_1}+\cdots+\alpha_ux^{s_u})}.\]
In this way, the weight distribution of cyclic code $\cC$ can be
derived from the explicit evaluating of the exponential sums
\[S(\alpha_1,\cdots,\alpha_u)\quad(\alpha_1,\cdots,\alpha_u\in \bF_q).\]

 Let $h_1(x)$, $h_2(x)$ and $h_{3}(x)$
be the minimal polynomials of $\pi^{-1},\pi^{-(2^k+1)}$ and
$\pi^{-{(2^m+1)}}$ over $\bF_{2}$ respectively. Then
\begin{equation}\label{deg}
\mathrm{deg}\,h_i(x)=n\; \text{for}\; i=1,2\;\text{and}\;
\mathrm{deg}\,h_3(x)=m
\end{equation}

 Let
$\cC_1$ and $\cC_2$ be the binary cyclic codes  with length $l=q-1$
and parity-check polynomials $h_2(x)h_3(x)$ and $h_1(x)h_2(x)h_3(x)$
respectively. It is a consequence that $\cC_1$($\cC_2$, resp.) is
the dual of the binary BCH code with designed distance $5$ ($7$,
resp.) whose zeroes include $\pi^{2^m+1}$ and $\pi^{2^k+1}$
($\pi^{2^m+1}$ ,$\pi^{2^k+1}$ and $\pi$, resp.). From (\ref{deg}),
we know that the dimensions of $\cC_1$ and $\cC_2$ over $\bF_{p^t}$
are $3n/2$ and $5n/2$ respectively.

For $\al\in \bF_{2^m},(\be,\ga)\in \bF_q^2$,  define the
exponential sums
\begin{equation}\label{def T}
T(\al,\be)=\sum\limits_{x\in \bF_q}(-1)^{\Tra_1^m (\alpha
x^{2^{m}+1})+\Tra_1^n(\beta x^{2^k+1})}
\end{equation}
and
\begin{equation}\label{def S}
S(\al,\be,\ga)=\sum\limits_{x\in \bF_q}(-1)^{\Tra_1^m (\alpha
x^{2^{m}+1})+\Tra_1^n(\beta x^{2^k+1}+\ga x)}.
\end{equation}
 Then the complete
weight distributions of $\cC_1$ and $\cC_2$ can be derived from the
explicit evaluation of $T(\al,\be)$ and $S(\al,\be,\ga)$.

 Another application of binary exponential sums is to obtain the cross correlation and auto-correlation distribution
 among binary sequences.
 Let $\mathcal{F}$ be a collection of binary
m-sequences of period $q-1$ defined by

\[\mathcal{F}=\left\{\left\{a_i(t)\right\}_{i=0}^{q-2}|\,0\leq i\leq L-1 \right\}\]

The \emph{correlation function} of $a_i$ and $a_j$ for a shift
$\tau$ is defined by

\[M_{{i},{j}}(\tau)=\sum\limits_{\lambda=0}^{q-2}(-1)^{a_i({\lambda})-a_j({\lambda+\tau})}\hspace{2cm}(0\leq \tau\leq
q-2).\]

Binary sequences with low cross correlation and auto-correlation are
widely used in Code Division Multiple Access(CDMA) spread
spectrum(see Paterson\cite{Part},  Simon, Omura and Scholtz\cite{Sim
Omu}). Pairs of binary m-sequences with few-valued auto and cross
correlations have been extensively studied for several decades, see
Canteaut, Charpin and Dobbertin \cite{Bluh}, Cusick and Dobbertin
\cite{Cus}, Ding, Helleseth and Lam\cite{Din Hel1}, Ding, Helleseth
and Martinsen\cite{Din Hel2}, Dobbertin, Felke, Helleseth and
Rosendahl \cite{Dob Fel},  Gold \cite{Gold}, Helleseth
\cite{Hell1},\cite{Hell2}, Helleseth, Kholosha and Ness\cite{Hel
Kho}, Helleseth and Kumar \cite{Hel Kum}, Hollmann and Xiang
\cite{Hol Xia}, Ness and Helleseth\cite{Nes Hel}, Niho\cite{Niho},
Rosendahl\cite{Rose}, Yu and Gong\cite{Yu Gon1}-\cite{Yu Gon2} and
references therein.

Define the collection of sequences

\[
  \calF_1=\left\{a_{\al,\be}=\left(a_{\al,\be}(\pi^{\lambda})\right)_{\lambda=0}^{q-2}\big{|}\,\al\in \bF_{2^m}, \be\in \bF_{q} \right\}
\]
where
  $a_{\al,\be}(\pi^{\lambda})=\Tra_1^m(\al \pi^{\lambda(2^m+1)})+\Tra_1^n(\be
  \pi^{\lambda(2^k+1)}+\pi^{\lambda})$.

 If $m/d$ or $k/d$ is even, define
\[
  \calF_2=\left\{a_{\be}=\left(a_{\be}(\pi^{\lambda})\right)_{\lambda=0}^{q-2}\big{|}\,\be=\pi^{i} \,\text{for}\;0\leq i\leq 2^m-2 \right\}
\]
where
  $a_{\be}(\pi^{\lambda})=\Tra_1^m(\pi^{\lambda(2^m+1)})+\Tra_1^n(\be
  \pi^{\lambda(2^k+1)})$.

If $k/d$ is even, define $\calF_3=\left\{a\right\} $ with
  $a=\left(\Tra_1^n(\pi^{\lambda(2^k+1)})\right)_{\lambda=0}^{q-2}$.

In this correspondence we will study the following collection of
m-sequences with period $q-1$
\begin{equation}\label{def F}
\calF=\left\{
\begin{array}{ll}
\calF_1, &\text{if}\; m/d\;\text{and}\; k/d\;\text{are both odd}\\[1mm]
\calF_1\cup \calF_2, &\text{if}\; m/d\;\text{is even}\\[1mm]
\calF_1\cup \calF_2\cup \calF_3, &\text{if}\;  k/d\;\text{is even.}\\[1mm]
\end{array}
\right.
\end{equation}

It is easy to verify that the sequences in $\calF$ are all cyclic
inequivalent and maximal with size $q-1$. Several particular cases
of the cyclic code $\cC_2$ or the related sequences collection
$\calF$ have been investigated, for instance:
\begin{itemize}
  \item The binary code $\cC_2$ with $k=m\pm 1$ is nothing but the
  classical Kasami code, see Kasami \cite{Kasa1}.
  \item As for the binary code $\cC_2$ with $k=1$, its minimal
  distance is obtained by Lahtonen \cite{Laht}, Moreno and Kumar
  \cite{Mor Kum}. Its weight distribution is determined eventually
  in van der Vlugt \cite{Vand2}.
  \item For the case $(k,n)=2$ if $m$ is odd, or $(k,n)=1$ if $m$ is even, the binary code $\cC_2$
  and the related family of generalized Kasami sequences have been
  studied, see Zeng, Liu and Hu \cite{Zen Liu}.
  \item In the case $p$ odd prime and $\gcd(m, k)=\gcd(m+k,2k)=d$
  being odd, the weight distribution of $\cC_2$ and correlation
  distribution of corresponding sequences have been fully
  determined, see Zeng, Li and Hu \cite{Zen Li}.
\end{itemize}

 This paper is presented as follows. In Section 2 we introduce
some preliminaries and give auxiliary results. In Section 3 we
will give the value distribution of $T(\al,\be)$ for $\al\in
\bF_{2^m},\be\in \bF_q$ and the weight distribution of $\cC_1$. In
Section 3 we will determine the value distribution of
$S(\al,\be,\ga)$,  the correlation distribution among the
sequences in $\calF$, and then the weight distribution of $\cC_2$.
Most proofs of lemmas and theorems are presented in several
appendices. The main tools are quadratic form theory over finite
fields of characteristic 2, some moment identities on
$T(\alpha,\beta)$ and a class of Artin-Schreier curves.

\section{Preliminaries}

\quad We follow the notations in Section 1. The first machinery to
determine the values of exponential sums $T(\alpha,\beta)$
$(\alpha\in \bF_{p^m},\beta\in \bF_q)$ defined in (\ref{def T}) is
quadratic form theory over $\bF_{q_0}$.

 Let $H$
be an $s\times s$ matrix over $\bF_{q_0}$. For the quadratic form
\begin{equation}\label{qua for}
F:\bF_{q_0}^s\ra \bF_{q_0},\quad F(x)=XHX^T\quad
(X=(x_1,\cdots,x_s)\in \bF_{q_0}^s),
\end{equation}
define $r_F$ of $F$ to be the rank of the skew-symmetric matrix
$H+H^T$. Then $r_F$ is even.

\begin{lemma}\label{qua}

For the quadratic form $F(X)=XHX^T$ defined in (\ref{qua for}),

\[\sum\limits_{X\in\bF_{q_0}^s}\zeta_p^{\Tra_1^{d}(F(X))}=\pm
q_0^{s-\frac{r_F}{2}}\; \text{or}\; 0\] Moreover, if $r_F=s$, then
\[\sum\limits_{X\in\bF_{q_0}^s}\zeta_p^{\Tra_1^{d}(F(X))}=\pm
q_0^{\frac{s}{2}}\]
\end{lemma}

\begin{proof}
We can calculate
\[{\setlength\arraycolsep{2pt}
\begin{array}{ll}&\left(\sum\limits_{X\in\bF_{q_0}^s}(-1)^{\Tra_{1}^{d}(F(X))}\right)^2
              =\sum\limits_{X,Z\in
\bF_{q_0}^s}(-1)^{\Tra_{1}^{d}\left(XHX^T+(X+Z)H(X+Z)^T\right)}\\[4mm]
              &\qquad\qquad=\sum\limits_{Z\in
\bF_{q_0}^s}(-1)^{\Tra_{1}^{d}(ZHZ^T)}\sum\limits_{X\in
\bF_{q_0}^s}(-1)^{\Tra_{1}^{d}\left(Z(H+H^T)X^T\right)}
\end{array}
}
\]
The inner sum is zero unless $Z(H+H^T)=0$. Define
\[\mathcal{Z}=\left\{Z\in \bF_{q_0}^s\,|\, Z(H+H^T)=0\right\}.\]
Then the map
\begin{equation}\label{map}
\begin{array}{rcl}
\eta:\mathcal{Z}&\longrightarrow& \bF_2\\[2mm]
Z&\mapsto& \Tra_1^d(ZHZ^T)
\end{array}
\end{equation}
is an additive group homomorphism.

If $\eta$ is surjective, then there are exactly one half $z\in
\mathcal{Z}$ mapping to $0$ and $1$ respectively. Hence
$\sum\limits_{X\in\bF_{q_0}^s}(-1)^{\Tra_{1}^{d}(F(X))}=0$.
Otherwise $\mathrm{Im}(\eta)=\left\{0\right\}$ and
$\left(\sum\limits_{X\in\bF_{q_0}^s}(-1)^{\Tra_{1}^{d}(F(X))}\right)^2=
q_0^{s}\cdot q_0^{s-r_F}$. Hence
$\sum\limits_{X\in\bF_{q_0}^s}(-1)^{\Tra_{1}^{d}(F(X))}=\pm
q_0^{s-\frac{r_F}{2}}$.

If $r_F=s$, then $\mathcal{Z}=\{0\}$ and
$\mathrm{Im}(\eta)=\left\{0\right\}$. Therefore
$\sum\limits_{X\in\bF_{q_0}^s}(-1)^{\Tra_{1}^{d}(F(X))}=\pm
q_0^{\frac{s}{2}}$.

\end{proof}

The following result, which has been proven in \cite{Lid Nie},
Chap. 6, will be used in Section 4.
\begin{lemma}\label{det gamma}
For the fixed quadratic form defined in  (\ref{qua for}), the value
distribution of
$\sum\limits_{X\in\bF_{q_0}^s}(-1)^{\Tra_1^{d}(F(X)+AX^T)}$ when $A$
runs through $\bF_{q_0}^s$ is shown as following
\[
\begin{array}{ccc}
value & \qquad\qquad\qquad&multiplicity \\[2mm]
0&\qquad\qquad\qquad&q_0^s-q_0^{r_F}\\[2mm]
q_0^{s-\frac{r_F}{2}}&\qquad\qquad\qquad &\frac{1}{2}(q_0^{r_F}+q_0^{\frac{r_F}{2}})\\[2mm]
-q_0^{s-\frac{r_F}{2}}&\qquad\qquad\qquad &\frac{1}{2}(q_0^{r_F}-q_0^{\frac{r_F}{2}})\\[2mm]
\end{array}
\]
\end{lemma}

Since $s=n/d$, the field $\bF_q$ is a vector space over $\bF_{q_0}$
with dimension $s$. We fix a basis $v_1,\cdots,v_s$ of $\bF_q$ over
$\bF_{q_0}$. Then each $x\in \bF_q$ can be uniquely expressed as
\[x=x_1v_1+\cdots+x_sv_s\quad (x_i\in \bF_{q_0}).\]
Thus we have the following $\bF_{q_0}$-linear isomorphism:
\[\bF_q\xrightarrow{\sim}\bF_{q_0}^s,\quad x=x_1v_1+\cdots+x_sv_s\mapsto
X=(x_1,\cdots,x_s).\] With this isomorphism, a function $f:\bF_q\ra
\bF_{q_0}$ induces a function $F:\bF_{q_0}^s\ra \bF_{q_0}$ where for
$X=(x_1,\cdots,x_s)\in \bF_{q_0}^s, F(X)=f(x)$ with
$x=x_1v_1+\cdots+x_sv_s$. In this way, function
$f(x)=\Tra_{d}^n(\gamma x)$ for $\gamma\in \bF_q$ induces a linear
form \begin{equation} F(X)=\Tra_{d}^n(\gamma
x)=\sum\limits_{i=1}^{s}\Tra_{d}^n(\gamma v_i)x_i=A_{\ga}X^T
\end{equation}\label{def A_gamma}
 where $A_{\ga}=\left(\Tra_{d}^n(\gamma
v_1),\cdots,\Tra_{d}^n(\gamma v_s)\right),$
 and
$f_{\alpha,\beta}(x)=\Tra_{d}^m(\alpha x^{p^m+1})+\Tra_d^n(\beta
x^{p^k+1})$  for $\al\in \bF_{p^m}$, $\be\in \bF_q$ induces a
quadratic form

\begin{eqnarray}\label{def H_al be}
F_{\alpha,\beta}(X)=XH_{\alpha,\beta}X^T
\end{eqnarray}

From Lemma \ref{qua},  for $\alpha\in \bF_{2^m},(\beta,\ga)\in
\bF_q^2$, in order to determine the values of
\[T(\alpha,\beta)=\sum\limits_{x\in \bF_q}(-1)^{\Tra_1^m (\alpha
x^{2^{m}+1})+\Tra_1^n(\beta x^{2^k+1})}=\sum\limits_{X\in
\bF_{q_0}^s}(-1)^{\Tra_1^{d}\left(XH_{\alpha,\beta}X^T\right)}\]
and
\[S(\alpha,\beta,\ga)=\sum\limits_{x\in \bF_q}(-1)^{\Tra_1^m (\alpha
x^{2^{m}+1})+\Tra_1^n(\beta x^{2^k+1}+\ga x)}=\sum\limits_{X\in
\bF_{q_0}^s}(-1)^{\Tra_1^{d}\left(XH_{\alpha,\beta}X^T+A_{\ga}X^T\right)},\]
we need to determine the rank of $H_{\alpha,\beta}+H_{\al,\be}^T$
over $\bF_{q_0}$.

Define $d'=\gcd(m+k,2k)$. Then an easy observation shows
\begin{equation}\label{rel d d'}
d'=\left\{
\begin{array}{ll}
2d, & \text{if}\; m/d\;\text{and}\; k/d \;\text{are both odd};\\[1mm]
d, &\text{otherwise.}
\end{array}
\right.
\end{equation}

Special case of the subsequent result has been proven in \cite{Zen
Liu}.
\begin{lemma}\label{rank}
For $(\alpha,\beta)\in \bF_{p^m}\times \bF_q\backslash\{(0,0)\}$,
let $r_{\alpha,\beta}$ be the rank of
$H_{\alpha,\beta}+H_{\alpha,\beta}^T$.  Then we have
\begin{itemize}
  \item[(i).] if $d'=d$, then the possible values of $r_{\al,\be}$
  are $s$, $s-2$.
  \item[(ii).] if $d'=2d$, then the possible values of $r_{\al,\be}$
  are $s$, $s-2$, $s-4$.
\end{itemize}

Moreover, let $n_i$ be the number of $({\alpha,\beta})$ with
$r_{\alpha,\beta}=s-i$. In the case $d'=d$, we have
    \[
    \begin{array}{ll}
    &n_0=\left(2^{n+2d}-2^{n+d}-2^n+2^{m+2d}-2^{m+d}+2^{2d}\right)(2^m-1)\large{/}(2^{2d}-1)\\[2mm]
    &n_2=(2^{m+d}-1)(2^n-1)\large{/}(2^{2d}-1).
    \end{array}
    \]
\end{lemma}
We need to introduce some results to prove Lemma \ref{rank}.

\begin{lemma}\label{num solution}(see Bluher \cite{Bluh}, Theorem 5.4 and 5.6)
Let $g(z)=z^{p^h+1}-bz+b$ with $b\in \bF_{p^l}^*$ and $e=\gcd(h,l)$.
Then we have
\begin{itemize}
    \item[(i).]the number of the solutions to $g(z)=0$ in
$\bF_{p^l}$ is $0$, $1$, $2$ or $p^{e}+1$.
    \item[(ii).]if $z_0$ is the
unique solution in $\bF_{p^l}^*$, then
$(z_0-1)^{\frac{p^l-1}{p^{e}-1}}=1.$
    \item[(iii).]denote by $N_i$ the number of $b\in \bF_{p^l}^*$
    such that $g(z)=0$ has exactly $i$ roots in $\bF_{2^l}$. Then
    we get
       \begin{itemize}
        \item if $l/e$ is even, then
        \[N_0=\frac{2^{l+e}-2^e}{2(2^e+1)},\quad N_1=2^{l-e},\quad N_2=\frac{(2^e-2)(2^l-1)}{2(2^e-1)},\quad N_{2^e+1}=\frac{2^{l-e}-2^e}{2^{2e}-1}.\]
        \item if $l/e$ is odd, then
        \[N_0=\frac{2^{l+e}+2^e}{2(2^e+1)},\quad N_1=2^{l-e}-1,\quad N_2=\frac{(2^e-2)(2^l-1)}{2(2^e-1)},\quad N_{2^e+1}=\frac{2^{l-e}-1}{2^{2e}-1}.\]
       \end{itemize}
\end{itemize}

\end{lemma}

The following lemma has been proven in \cite{Bluh},\cite{Zen Li} and
\cite{Zen Liu}. We will repeat part of the proof for
self-containing.
\begin{lemma}\label{psi}
Let $\psi_{\al,\be}(z)=\be^{2^{n-k}} z^{2^{m-k}+1}+\al z+\be$ with
$\al\in \bF_{2^m}^*, \be\in \bF_q^*$. Then
\begin{itemize}
  \item[(i).] $\psi_{\al,\be}(z)=0$ has either $0,1,2$ or $2^{d'}+1$
  solutions in $\bF_q$.
  \item[(ii).]  If $z_1, z_2$ are two solutions of
  $\psi_{\al,\be}(z)=0$ in $\bF_q$, then $(z_1z_2)^{\frac{q-1}{2^d-1}}=1$.
  \item[(iii).] If $\psi_{\al,\be}(z)=0$ has $2^{d'}+1$
  solutions in $\bF_q$, then for any two solutions $z_1$ and $z_2$, we
  have $(z_1/z_2)^{\frac{q-1}{2^{d'}-1}}=1$.
  \item[(iv).] If $\psi_{\al,\be}(z)=0$ has exactly one
  solution in $\bF_q$, say $z_0$, then $z_0^{\frac{q-1}{2^d-1}}=1$.
\end{itemize}
\end{lemma}
\begin{proof}
\begin{itemize}
  \item[(i).] By scaling $y=\frac{\al}{\be} z$ and
$b=\frac{\al^{2^{m-k}+1}}{\beta^{2^{m-k}(2^m+1)}}$, we can rewrite
the equation $\psi_{\al,\be}(z)=0$ as
\begin{equation}\label{standard equ}
g(y)=y^{2^{m-k}+1}+by+b=0.
\end{equation}
Since $b\in \bF_{2^m}^*\subseteq\bF_q^*$ and $\gcd(m-k,
n)=\gcd(m-k,2k)=d'$, then the result follows from Lemma \ref{num
solution} (i).
  \item[(ii).] See \cite{Zen Li}, Prop.1 (2).
  \item[(iii).] Denote by $y_i=\frac{\al}{\be} z_i$ for $i=1,2$. Since $\gcd(n,m-k)=d'$, from \cite{Bluh}, Theorem 4.6 (iv) we get $(y_1/y_2)^{\frac{q-1}{2^{d'}-1}}=1$
  which is equivalent to
  $(z_1/z_2)^{\frac{q-1}{2^{d'}-1}}=1$.
  \item[(iv).] See \cite{Zen Li}, Prop.1 (3).
\end{itemize}
\end{proof}
\begin{remark}
\begin{itemize}
    \item[(i).] If $\psi_{\al,\be}(z)=0$ has exactly one or $2^{d'}+1$
  solutions in $\bF_q$, then each solution is a $(2^d-1)$-th power
  in $\bF_q$.
    \item[(ii).] If $\psi_{\al,\be}(z)=0$ has exactly two solutions in
    $\bF_q$, then none or both of them are $(2^d-1)$-th powers
  in $\bF_q$.
\end{itemize}
\end{remark}

{\it \textbf{Proof of Lemma \ref{rank}}}: (i). For
$Y=(y_1,\cdots,y_s)\in \bF_{q_0}^s$, $y=y_1v_1+\cdots+y_sv_s\in
\bF_q$, we know that
\begin{equation}\label{bil form1}
F_{\alpha,\beta}(X+Y)-F_{\alpha,\beta}(X)-F_{\alpha,\beta}(Y)=X\left(H_{\alpha,\beta}+H_{\alpha,\beta}^T\right)Y^T
\end{equation}
is equal to
\begin{equation}\label{bil form2}
f_{\alpha,\beta}(x+y)-f_{\alpha,\beta}(x)-f_{\alpha,\beta}(y)=\Tra_{d}^n\left(y(\alpha
x^{2^{m}}+\beta x^{2^k}+\be^{2^{n-k}} x^{2^{n-k}})\right).
\end{equation}
 Let
\begin{equation}\label{def phi}
\phi_{\al,\be}(x)=\alpha x^{2^{m}}+\beta x^{2^k}+\be^{2^{n-k}}
x^{2^{n-k}}.
\end{equation}
 Therefore,
\[{\setlength\arraycolsep{2pt}
\begin{array}{lcl}
r_{\al,\be}=r& \Leftrightarrow&\text{the number of common solutions of}\;X\left(H_{\alpha,\beta}+H_{\alpha,\beta}^T\right)Y^T=0\;\text{for all}\;Y\in \bF_{q_0}^s\;\text{is}\; q_0^{s-r}, \\[2mm]
& \Leftrightarrow&\text{the number of common solutions of}\;\Tra_{d}^n\left(y\cdot\phi_{\al,\be}(x)\right)=0\;\text{for all}\;y\in \bF_q\;\text{is}\; q_0^{s-r}, \\[2mm]
&\Leftrightarrow&\phi_{\al,\be}(x)=0\;\text{has}\; q_0^{s-r}\;
\text{solutions in}\; \bF_q.
\end{array}
}
\]

Since $\phi_{\al,\be}(x)$ is a $2^d$-linearized polynomial,  then
the set of the zeroes to $\phi_{\al,\be}(x)=0$ in $\bF_{q}$, say
$V$, forms an $\bF_{2^d}$-vector space.

If $\al=0$ and $\be\neq 0$, $\phi_{\al,\be}(x)=0$ becomes $\be
x^{2^k}+\be^{2^{n-k}}x^{2^{n-k}}=0$ and then $\be
^{2^k}x^{2^{2k}}+\be x=0$. In this case (\ref{def phi}) has $1$ or
$2^{d'}$ solutions according to $\be^{1-2^k}$ is $(2^{2d}-1)$-th
power in $\bF_q$ or not. Hence $r_{0,\be}=s$ or $s-2$. More
precisely, in the case $m/d$ is even, $\be^{1-2^k}$ is
$(2^{2d}-1)$-th power in $\bF_q^*$ if and only if $\be$ is a
$(p^d+1)$-th power in $\bF_q^*$. Hence the numbers of $\be\in
\bF_q^*$ such that $r_{0,\be}=s-2$ and $s$ are exactly
$\frac{2^n-1}{2^d+1}$ and $\frac{2^d(2^n-1)}{2^d+1}$ respectively.
In the case $k/d$ is even, $\be^{1-2^k}$ is always $(2^{2d}-1)$-th
power in $\bF_q^*$ which follows that $r_{0,\be}=s-2$ for $\be\in
\bF_q^*$. If $\al\neq 0$ and $\be=0$, then $\phi_{\al,0}(x)=0$ has
unique solution $x=0$ and as a consequence $r_{\al,0}=s$.

In the following we assume $\al\be\neq 0$, we need to consider the
nonzero solutions of $\phi_{\al,\be}(x)=0$. By substituting
$z=x^{2^k(p^{m-k}-1)}$ we get
\begin{equation}\label{def psi}
\psi_{\al,\be}(z)=\be^{2^{n-k}}z^{2^{m-k}+1}+\al z+\be=0.
\end{equation}

From Lemma \ref{num solution}, $\psi_{\al,\be}(z)=0$ has either
$0,1,2$ or $2^{d'}+1$ solutions in $\bF_q$.  In the case $d'=d$, by
Lemma \ref{psi} and its Remark,
\begin{itemize}
    \item if $\psi_{\al,\be}(z)=0$ has no solution in $\bF_q$,
    then $\phi_{\al,\be}(x)=0$ has unique solution $x=0$ in
    $\bF_q$ and $r_{\al,\be}=s$.
    \item if $\psi_{\al,\be}(z)=0$ has exactly one in $\bF_q$
    which is also a $(2^d-1)$-th power, then $r_{\al,\be}=s-1$.
    But $s-1$ is not even that leads to a contradiction.
    \item if $\psi_{\al,\be}(z)=0$ has two solutions in $\bF_q$,
    then $\phi_{\al,\be}(x)=0$ has one or $2(2^d-1)+1$ solutions in
    $\bF_q$. Note that $2^{d+1}-1$ is not a $2^d$-th power
which is impossible. Hence $r_{\al,\be}=s$.
    \item if $\psi_{\al,\be}(z)=0$ has $2^d+1$ solutions in
    $\bF_q$, then they are all $(2^d-1)$-th power. Therefore $\phi_{\al,\be}(x)=0$ has
    $(2^d+1)(2^d-1)+1=2^{2d}$ solutions in $\bF_q$ and
    $r_{\al,\be}=s-2$.
\end{itemize}

Again by Lemma \ref{num solution}, $\psi_{\al,\be}(z)=0$ has either
$0,1,2$ or $2^{d}+1$ solutions in $\bF_{2^m}$.  Note that the
solutions of $\psi_{\al,\be}(z)=0$ have one-to-one correspondence
with those of $g(y)=0$ by substituting $y=\frac{\al}{\be}z$.  If
$y_0\in \bF_q\big{\backslash}\bF_{2^m}$ satisfying $g(y_0)=0$, then
$g(y_0^{2^m})=0$. Hence
\begin{itemize}
    \item if $g(y)=0$ has no or two solutions in
    $\bF_{2^m}$, then $g(y)=0$ and $\psi_{\al,\be}(z)=0$ has no or two solutions in
    $\bF_{q}$ and $r_{\al,\be}=s$.
    \item $\psi_{\al,\be}(z)=0$ has one or $2^d+1$ solutions in
    $\bF_{2^m}$, then $g(y)=0$ and $\psi_{\al,\be}(z)=0$ has $2^{d}+1$ solutions in
    $\bF_{q}$ and $r_{\al,\be}=s-2$.
\end{itemize}
For fixed $b$ and $\al\in \bF_{2^m}^*$, the number of $\be\in
\bF_{q}^*$ satisfying
$b=\frac{\al^{2^{m-k}+1}}{\beta^{2^{m-k}(p^m+1)}}$ is $2^m+1$.
Applying Lemma \ref{num solution} with $l=m$ and $h=m+k$ we get that
\begin{itemize}
    \item if $m/d$ is even, then
    \[
    \begin{array}{ll}
    &n_0=\left(\frac{2^{m+d}-2^d}{2(2^d+1)}+\frac{(2^d-2)(2^m-1)}{2(2^d-1)}\right)\cdot (2^m-1)\cdot
    (2^m+1)+\frac{2^d(2^n-1)}{2^d+1}+(2^m-1)\\[2mm]
    &\quad
    =\left(2^{n+2d}-2^{n+d}-2^n+2^{m+2d}-2^{m+d}+2^{2d}\right)(2^m-1)\large{/}(2^{2d}-1).
    \end{array}
    \]
    \item if $k/d$ is even, then
    \[
    \begin{array}{ll}
    &n_0=\left(\frac{2^{m+d}+2^d}{2(2^d+1)}+\frac{(2^d-2)(2^m-1)}{2(2^d-1)}\right)\cdot (2^m-1)\cdot
    (2^m+1)+(2^m-1)\\[2mm]
    &\quad
    =\left(2^{n+2d}-2^{n+d}-2^n+2^{m+2d}-2^{m+d}+2^{2d}\right)(2^m-1)\large{/}(2^{2d}-1).
    \end{array}
    \]
\end{itemize}
By $n_0+n_2=2^{3m}-1$ we get the result.

In the case $d'=2d$,  a similar argument gives the result. $\square$

 In order to determine the multiplicity of each value of
$T(\al,\be)$
 for $\al\in \bF_{p^m},\be\in \bF_q$, we need the
following result on moments of $T(\al,\be)$.
\begin{lemma}\label{moment}
For the exponential sum $T(\al,\be)$,
\[\begin{array}{ll}&(i). \;\;\sum\limits_{\al\in \bF_{2^m},\be\in
\bF_q}T(\al,\be)=2^{3m};\\[4mm]
                   &(ii).
                   \sum\limits_{\al\in \bF_{2^m},\be\in
\bF_q}T(\al,\be)^2=\left\{\begin{array}{ll} 2^{5m} &\text{if}\;
d'=d\\[2mm]
2^{3m}(2^{n+d}+2^{n}-2^d)&\text{if}\; d'=2d;
\end{array}\right.\\[4mm]
                         &(iii).
                   \sum\limits_{\al\in \bF_{2^m},\be\in
\bF_q}T(\al,\be)^3=\left\{\begin{array}{ll}
2^{3m}(2^{n+d}+2^{n}-2^d) &\text{if}\;
d'=d\\[2mm]
2^{3m}(2^{n+3d}+2^n-2^{3d})&\text{if}\; d'=2d;
\end{array}\right.
\end{array}\]
\end{lemma}
\begin{proof}
(i). We observe that
\[ {\setlength\arraycolsep{2pt}
\begin{array}{ll}
&\sum\limits_{\al\in \bF_{p^m},\be\in
\bF_q}T(\al,\be)=\sum\limits_{\al\in
\bF_{2^m},\be\in\bF_q}\sum\limits_{x\in
\bF_q}(-1)^{\Tra_1^m(\al x^{2^m+1})+\Tra_1^n(\be x^{2^k+1})}\\[3mm]
&\quad\quad=\sum\limits_{x\in\bF_q}\sum\limits_{\al\in
\bF_{2^m}}(-1)^{\Tra_1^m(\al x^{2^m+1})}\sum\limits_{\be\in
\bF_q}(-1)^{\Tra_1^n(\be
x^{2^k+1})}=q\cdot\sum\limits_{\stackrel{\al\in
\bF_{2^m}}{x=0}}(-1)^{\Tra_1^m(\al x^{2^m+1})}=2^{3m}.
\end{array}
}
\]
(ii). We can calculate
\[
{ \setlength\arraycolsep{2pt}
\begin{array}{lll}
\sum\limits_{\al\in
\bF_{2^m},\be\in\bF_q}T(\al,\be)^2&=&\sum\limits_{x,y\in
\bF_q}\sum\limits_{\al\in
\bF_{2^m}}(-1)^{\Tra_1^m\left(\al\left(x^{2^m+1}+y^{2^m+1}\right)\right)}\sum\limits_{\be\in
\bF_q}(-1)^{\Tra_1^n\left(\be\left(x^{2^k+1}+y^{2^k+1}\right)\right)}\\[2mm]
&=&M_2\cdot 2^{3m}\end{array}}
\] where
$M_2$ is the number of solutions to the equation
\begin{eqnarray}\label{def 2nd}
\left\{
\begin{array}{ll}
 x^{2^m+1}+y^{2^m+1}=0&\\[2mm]
 x^{2^k+1}+y^{2^k+1}=0&
 \end{array}
 \right.
\end{eqnarray}

If $xy=0$ satisfying (\ref{def 2nd}), then $x=y=0$. Otherwise
$(x/y)^{2^m+1}=(x/y)^{2^k+1}=1$ which yields that
$(x/y)^{2^{m-k}-1}=1$. Denote by $x=ty$.  Since $\gcd(m-k,n)=d'$,
then $t\in \bF_{2^{d'}}^*$.
\begin{itemize}
  \item If $d'=d$, then $t\in \bF_{2^d}^*$ and (\ref{def 2nd}) is
  equivalent to $x=y$. Hence $M_2=1+(q-1)=q$.
  \item If $d'=2d$, then by (\ref{rel d d'}) we get (\ref{def 2nd}) is equivalent to
  $x^{2^d+1}=y^{2^d+1}$. Then we have $t^{2^d+1}=1$ which has
  $2^d+1$ solutions in $\bF_{2^{d'}}^*$. Therefore
  \[M_2=(2^d+1)(2^n-1)+1=2^{n+d}+2^n-2^d.\]
\end{itemize}

(iii). We have
\[
{ \setlength\arraycolsep{2pt}
\begin{array}{lll}
\sum\limits_{\al\in
\bF_{2^m},\be\in\bF_q}T(\al,\be)^3&=&\sum\limits_{x,y,z\in
\bF_q}\sum\limits_{\al\in
\bF_{2^m}}(-1)^{\Tra_1^m\left(\al\left(x^{2^m+1}+y^{2^m+1}+z^{2^m+1}\right)\right)}\sum\limits_{\be\in
\bF_q}(-1)^{\Tra_1^n\left(\be\left(x^{2^k+1}+y^{2^k+1}+z^{2^k+1}\right)\right)}\\[2mm]
&=&M_3\cdot 2^{3m}\end{array}}
\] where
$M_3$ is the number of solutions to the equation
\begin{eqnarray}\label{def 3rd}
\left\{
\begin{array}{ll}
 x^{2^m+1}+y^{2^m+1}+z^{2^m+1}=0&\\[2mm]
 x^{2^k+1}+y^{2^k+1}+z^{2^k+1}=0&
 \end{array}
 \right.
\end{eqnarray}
In the case $xyz=0$, we may assume $z=0$. Then (\ref{def 3rd}) has
$M_2$ solutions. Hence (\ref{def 3rd}) has $3M_2-2$ solutions
satisfying $xyz=0$.

In the case $xyz\neq 0$. Assume $z=1$. By (\ref{def 3rd}) we get
\[
\left(x^{2^k+1}+1\right)^{2^m+1}=\left(y^{2^k+1}\right)^{2^m+1}=\left(y^{2^m+1}\right)^{2^k+1}=\left(x^{2^m+1}+1\right)^{2^k+1}.
\]
Therefore we have Therefore we get
\[(x^{2^{m+k}}+x)(x^{2^m}+x^{2^k})=0.\]
Since $\gcd(m+k,n)=\gcd(m-k,n)=d'$, we have $x\in\bF_{2^{d'}}$.
\begin{itemize}
    \item If $d'=d$, then (\ref{def 3rd}) with $z=1$ reduces to
    $x^2+y^2+1=0$, i.e. $x+y+1=0$ which has $2^d-2$ solutions in
    $\bF_{2^d}^2$ satisfying $xy\neq 0$. Hence
    \[M_3=3M_2-2+(2^d-2)\cdot(2^n-1)=2^{n+d}+2^n-2^d.\]
    \item If $d'=2d$, then (\ref{def 3rd}) with $z=1$ reduces to
\begin{equation}\label{Her}
    x^{2^d+1}+y^{2^d+1}+1=0
    \end{equation}
 with $(x,y)\in \bF_{2^{2d}}^*\times \bF_{2^{2d}}^*$.
 Note that (\ref{Her}) defines a Hermitian curve on $\bF_{2^{2d}}$ which
 has $2^{3d}+1$ rational (projective) points (see \cite{Stic}).
 Since (\ref{Her}) has $2^{d}+1$ infinite points: $(1:t_1:0)$ with $2^{2^d+1}=1$, and $2(2^d+1)$ solutions satisfying $xy=0$: $(t_1,0)$,
 $(0,t_2)$ with $t_1^{2^d+1}=t_2^{2^d+1}=1$, we get
 \[M_3=3M_2-2+\left((2^{3d}+1)-3(2^d+1)\right)\cdot (2^n-1)=2^{n+3d}+2^n-2^{3d}.\]
\end{itemize}

\end{proof}
\begin{remark}
For the case $d'=2d$, $\sum\limits_{\al\in \bF_{2^m},\be\in
\bF_q}T(\al,\be)^3$ can also be determined, but we do not need this
result.
\end{remark}

In the case $d'=2d$, we could determine the explicit values of
$T(\al,\be)$. To this end we need to study a class of Artin-Schreier
curves. A similar technique has been employed in Coulter
\cite{Coul}, Theorem 5.2.

\begin{lemma}\label{Artin}Suppose $(\al',\be)\in
\bF_q^2\big{\backslash}\{0,0\}$ such that $\Tra_m^n(\al')=\al$ and
$d'=2d$. Let $N$ be the number of $\bF_q$-rational (affine) points
on the curve
\begin{equation}\label{Artin Sch}
\al' x^{2^m+1}+\be x^{2^k+1}=y^{2^d}+y.
\end{equation}
 Then
\[N=2^n+(2^d-1)\cdot T(\al,\be).\]
\end{lemma}
\begin{proof}
 We get that
\[
\begin{array}{rcl}
qN&=&\sum\limits_{\om\in \bF_q}\sum\limits_{x,y\in
\bF_q}(-1)^{\Tra_1^n\left(\om\left(\al' x^{2^m+1}+\be x^{2^k+1}+y^{2^d}+y\right)\right)}\\[2mm]
&=&q^2+\sum\limits_{\om\in \bF_q^*}\sum\limits_{x\in
\bF_q}(-1)^{\Tra_1^n\left(\om\left(\al' x^{2^m+1}+\be
x^{2^k+1}\right)\right)} \sum\limits_{y\in
\bF_q}(-1)^{\Tra_1^n\left(y^{2^d}\left(\om^{2^d}+\om\right)\right)}\\[2mm]
&=&q^2+q\sum\limits_{\om\in \bF_{q_0}^*}\sum\limits_{x\in
\bF_q}(-1)^{\Tra_1^n\left(\om\left(\al' x^{2^m+1}+\be
x^{2^k+1}\right)\right)}\\[2mm]
&=&q^2+q\sum\limits_{\om\in \bF_{q_0}^*}\sum\limits_{x\in
\bF_q}(-1)^{\Tra_1^m\left(\om\al
x^{2^m+1}\right)+\Tra_1^n\left(\om\be x^{2^k+1}\right)}\\[2mm]
&=&q^2+q\sum\limits_{\om\in \bF_{q_0}^*}T(\om\al,\om\be)
\end{array}
\]
where the 3-rd equality follows from that the inner sum is zero
unless $\om^{2^d}+\om=0$, i.e. $\om\in \bF_{q_0}$ and the 4-th
equality follows from $\om \al' x^{p^m+1}\in \bF_{p^m}$.

For any $\om\in \bF_{q_0}^*$,  choose $t\in \bF_{2^{2d}}$ such that
$t^{2^d+1}=\om$. Then $t^{2^m+1}=t^{2^k+1}=\om$. As a consequence
$T({\om\al,\om\be})=T(\al,\be).$ Hence $N=2^n+(2^d-1)\cdot
T(\al,\be)$.
\end{proof}

Now we give an explicit evaluation of $T(\al,\be)$ in the case
$d'=2d$.
\begin{lemma}\label{reduce num}
Assumptions as in Lemma \ref{Artin}. Then
\[
T(\al,\be)=\left\{
\begin{array}{ll}
-2^m, &\text{if}\; r_{\al,\be}=s\\[2mm]
2^{m+d}, &\text{if}\; r_{\al,\be}=s-2\\[2mm]
-2^{m+2d}, &\text{if}\; r_{\al,\be}=s-4
\end{array}
\right.
\]
\end{lemma}
\begin{proof}
Consider the $\bF_q$-rational (affine) points on the Artin-Schreier
curve in Lemma \ref{Artin}. It is easy to verify that $(0,y)$ with
$y\in \bF_{q_0}$ are exactly the points on the curve with $x=0$. If
$(x,y)$ with $x\neq 0$ is a point on this curve, then so are $(t
x,t^{2^d+1}y)$ with $t^{2^{2d}-1}=1$ (note that $2^m+1\equiv
2^k+1\equiv 2^d+1\pmod {2^{2d}-1}$ since $m/d$ and $k/d$ are both
odd by (\ref{rel d d'})). In total, we have
\[2^n+(2^d-1)T(\al,\be)=N\equiv 2^d\pmod {2^{2d}-1}\]
which yields
\[T(\al,\be)\equiv 1\pmod{2^d+1}.\]

Obviously $T(\al,\be)\neq 0$. We only consider the case
$r_{\al,\be}=s$. The other cases are similar. In this case
$T(\al,\be)=\pm
 2^m$. Assume $T(\al,\be)=2^m$.
 Then ${2^d+1}\mid 2^m-1$ which contradicts to $m/d$ is odd. Therefore $T(\al,\be)=-2^m$.
\end{proof}
\begin{remark}
Applying Lemma \ref{reduce num} to Lemma \ref{Artin}, we could
determine the number of rational points on the curve (\ref{Artin
Sch}).
\end{remark}

\section{Exponential Sums $T(\al,\be)$ and Cyclic Code $\cC_1$}

\quad Recall $q_0^*=(-1)^{\frac{q_0-1}{2}}q_0$.
 In this section we prove the following results.

\begin{theo}\label{value dis T}
The value distribution of the multi-set
$\left\{T(\al,\be)\left|\al\in \bF_{2^m},\be\in
\bF_q\right.\right\}$ and the weight distribution of $\cC_1$ are
shown as following (Column 1 is the value of $T(\al,\be)$, Column
2 is the weight of $c(\al,\be)=\left(\Tra_1^m(\al
\pi^{i(2^m+1)})+\Tra_1^n(\be \pi^{i(2^k+1)})\right)_{i=0}^{q-2}$
and Column 3 is the corresponding multiplicity).

 (i). For the case $d'=d$,
\begin{center}
\begin{tabular}{|c|c|c|}
\hline
value &weight& multiplicity \\[2mm]
\hline $2^{m}$& $2^{n-1}-2^{m-1}$
&$\frac{2^{d-1}(2^m-1)(2^{n}+2^{m+1}+1)}{2^{d}+1} $
\\[2mm]
\hline
 $-2^{m}$& $2^{n-1}+2^{m-1}$&$
\frac{2^{d-1}(2^m-1)(2^{n}-2^{n-d+1}+1)}{2^{d}-1} $
\\[2mm]

\hline
$-2^{m+d}$&$2^{n-1}+2^{m+d-1}$&$\frac{(2^{m-d}-1)(2^n-1)}{2^{2d}-1}$\\[2mm]
\hline $0$&$2^{n-1}$& $2^{m-d}(2^n-1)$
\\[2mm]
\hline $2^n$ & $0$ & $1$\\[2mm]
\hline
\end{tabular}
\end{center}

(ii). For the case $d'=2d$,
\begin{center}
\begin{tabular}{|c|c|c|}
\hline
value &weight& multiplicity \\[2mm]
\hline
$-2^m$&$2^{n-1}+2^{m-1}$&$\frac{2^{3d}(2^m-1)(2^{n}-2^{n-2d}-2^{n-3d}+2^m-2^{m-d}+1)}{(2^d+1)(2^{2d}-1)}$
\\[2mm]
\hline
${2}^{m+d}$&$2^{n-1}-2^{m+d-1}$&$\frac{2^{d}(2^n-1)(2^m+2^{m-d}+2^{m-2d}+1)}{(2^d+1)^2}$
\\[2mm]
\hline
 $-{2}^{m+2d}$&$2^{n-1}-2^{m+2d-1}$&
$\frac{(2^{m-d}-1)(2^{n}-1)}{(2^d+1)(2^{2d}-1)}$
\\[2mm]

\hline $2^m$&$0$&$1$
\\[2mm]
\hline
\end{tabular}
\end{center}
\end{theo}
\begin{proof}
Define
\[N_{i}=\left\{(\al,\be)\in
\bF_{2^m}\times
\bF_q\backslash\{(0,0)\}\left|r_{\al,\be}=s-i\right.\right\}.
\] Then $n_i=\big{|}N_i\big{|}$.

 According to Lemma \ref{qua} (setting
$F(X)=XH_{\al,\be}X^T=\Tra_{d}^m(\al x^{2^m+1})+\Tra_{d}^n(\be
x^{2^k+1})$), we define that for $\varepsilon=\pm 1$ and $0\leq
i\leq s-1$,
\[N_{i,\varepsilon}=\left\{
(\al,\be)\in \bF_{2^m}\times
\bF_q\backslash\{(0,0)\}\left|T(\al,\be)=\veps\cdot
2^{m+\frac{id}{2}} \right.\right\}.\] and
$n_{i,\veps}=|N_{i,\veps}|$. Then by Lemma \ref{qua} we have
$N_0=N_{0,1}\bigcup N_{0,-1}$ and $n_0=n_{0,1}+n_{0,-1}$. But for
general $i$, $N_i\neq N_{i,1}\bigcup N_{i,-1}$ and $n_i\neq
n_{i,1}+n_{i,-1}$.

 Meanwhile we define $\omega$
to be the number of $(\al,\be)\in \bF_{2^m}\times
\bF_q\big{\backslash}\{(0,0)\}$ such that $T(\al,\be)=0$.

From (\ref{Wei}) we know that for each non-zero codeword
$c(\al,\be)=\left(c_0,\cdots,c_{l-1}\right)$ $(l=q-1,
c_i=\Tra_1^m(\al\pi^{(2^m+1)i})+\Tra_1^n(\be \pi^{(2^k+1)i}), 0\leq
i\leq l-1, \text{and}\; (\al,\be)\in \bF_{2^m}\times\bF_q)$, the
Hamming weight of $c(\al,\be)$ is
\begin{equation}\label{wei c}
w_H\left(c(\al,\be)\right)=2^{n-1}-\frac{1}{2}\cdot T(\al,\be).
\end{equation}

 (i). For the case $d'=d$, by Lemma \ref{qua} and Lemma \ref{rank}
 we know that the possible values of $T(\al,\be)$ for $(\al,\be)\in \bF_{2^m}\times \bF_q\big{\backslash}\{(0,0)\}$ is $0, \pm 2^m, \pm
 2^{m+d}$.
 Moreover from Lemma \ref{moment} we have
\begin{equation}\label{par sum1}
 \left(n_{0,1}-n_{0,-1}\right)+2^d\left(n_{2,1}-n_{2,-1}\right)=2^m(2^m-1)
 \end{equation}
\begin{equation}\label{par sum2}
 \left(n_{0,1}+n_{0,-1}\right)+2^{2d}\left(n_{2,1}+n_{2,-1}\right)=2^{n}(2^m-1)
 \end{equation}
\begin{equation}\label{par sum3}
 \left(n_{0,1}-n_{0,-1}\right)+2^{3d}\left(n_{2,1}-n_{2,-1}\right)=-2^{3m}+2^{n+d}+2^n-2^d.
 \end{equation}
In addition, by Lemma \ref{qua} and Lemma \ref{rank} we have
\begin{equation}\label{par sum0}
 \omega+\left(n_{0,1}+n_{0,-1}\right)+\left(n_{2,1}+n_{2,-1}\right)=2^{3m}-1.
 \end{equation}
and
   \begin{equation}\label{val n0'}
   n_{0,1}+n_{0,-1}=n_0=\frac{\left(2^{n+2d}-2^{n+d}-2^n+2^{m+2d}-2^{m+d}+2^{2d}\right)(2^m-1)}{2^{2d}-1}.
   \end{equation}
Combining (\ref{par sum1})-- (\ref{val n0'}) we get the result.

 (iii). For the case $d'=2d$, by Lemma \ref{reduce num} we have
 \begin{equation}\label{val all 0}
\omega=n_{0,1}=n_{2,-1}=n_{4,1}=0.
 \end{equation}
Combining Lemma \ref{rank}, Lemma \ref{moment}, Lemma (\ref{reduce
num}) and (\ref{val all 0}) we have
\begin{equation}\label{par sum02}
 n_{0,-1}+n_{2,1}+n_{4,-1}=2^{3m}-1
 \end{equation}

\begin{equation}\label{par sum12}
 -n_{0,-1}+2^d\cdot n_{2,1}-2^{2d}\cdot n_{4,-1}=2^m(2^{m}-1)
 \end{equation}

 \begin{equation}\label{par sum22}
 n_{0,-1}+2^{2d}\cdot n_{2,1}+2^{4d}\cdot
 n_{4,-1}=2^m(2^{n+d}+2^{n}-2^m-2^d).
 \end{equation}

Solving the system of equations consisting of (\ref{par
sum02})--(\ref{par sum22}) yields the result.

The weight distribution of $\cC_1$ is derived from the value
distribution of $T(\al,\be)$ and (\ref{wei c}).
\end{proof}

\section{Results on Correlation Distribution of Sequences and Cyclic Code $\cC_2$}

\quad Recall $\phi_{\al,\be}(x)$ in the proof of Lemma \ref{rank}
and $N_{i,\veps}$ in the proof of Theorem \ref{value dis T}.
Finally we will determine the value distribution of
$S(\al,\be,\ga)$, the correlation distribution among sequences in
$\calF$ defined in (\ref{def F}) and the weight distribution of
$\cC_2$ defined in Section 1.

The following lemma, which has been stated for the case $p$ odd in
\cite{Zen Li} Lemma 5, is also valid for $p=2$.

\begin{lemma}\label{q-2}
\begin{itemize}
\item[(i).] The sequences in $\calF$ are all maximal with length $q-1$.
  \item[(i).] For any given $\al\in \bF_{2^m}^*$, when $\be$ runs through
$\bF_q$, the distribution of $T(\al,\be)$ is the
same as $T(1,\be)$.
  \item[(ii).] For any given $\ga\in \bF_q^*$, when $(\al,\be)$ runs through
$\bF_{p^m}\times \bF_q$, the distribution of $S(\al,\be,\ga)$ is the
same as $S(\al,\be,1)$.
  \item[(iii).] Suppose $m/d$ or $k/d$ is even. If $\be$ runs through $\bF_q^*$,  then $C_{\be}$ runs through the sequences
  in $\calF$ exactly $(2^m+1)$ times.
\end{itemize}

\end{lemma}

 We are now ready to give the value
distribution of $S(\al,\be,\ga)$ and weight distribution of
$\cC_2$.

\begin{theo}\label{value dis S}
The value distribution of the multi-set
$\left\{S(\al,\be,\ga)\left|\al\in \bF_{2^m},(\be,\ga)\in
\bF_q^2\right.\right\}$ and the weight distribution of $\cC_2$ are
shown as following (Column 1 is the value of $S(\al,\be,\ga)$,
Column 2 is the weight of $c(\al,\be,\ga)=\left(\Tra_1^m(\al
\pi^{i(2^m+1)})+\Tra_1^n(\be \pi^{i(2^k+1)}+\ga
\pi^i)\right)_{i=0}^{q-2}$ and Column 3 is the corresponding
multiplicity).

 (i). For the case $d'=d$,
\begin{center}
\begin{tabular}{|c|c|c|}
\hline
value &weight& multiplicity \\[2mm]
\hline
$2^{m}$&$2^{n-1}-2^{m-1}$&$\frac{2^{m-1}(2^n-1)(2^{n+2d}-2^{n+d}-2^n+2^{m+2d}-2^{m+d}+2^{2d})}{2^{2d}-1}$
\\[2mm]
\hline
$-2^{m}$&$2^{n-1}+2^{m-1}$&$\frac{2^{m-1}(2^m-1)^2(2^{n+2d}-2^{n+d}-2^n+2^{m+2d}-2^{m+d}+2^{2d})}{2^{2d}-1}$
\\[2mm]

\hline
$2^{m+d}$&$2^{n-1}-2^{m+d-1}$&$\frac{2^{m-d-1}(2^{m-d}+1)(2^{m+d}-1)(2^n-1)}{2^{2d}-1}$\\[2mm]
\hline
$-2^{m+d}$&$2^{n-1}+2^{m+d-1}$&$\frac{2^{m-d-1}(2^{m-d}-1)(2^{m+d}-1)(2^n-1)}{2^{2d}-1}$\\[2mm]
\hline $0$&$2^{n-1}$&$(2^{3m-d}-2^{n-2d}+1)(2^n-1)$
\\[2mm]
\hline $2^n$ & $0$ & $1$\\[2mm]
\hline
\end{tabular}
\end{center}

(ii). For the case $d'=2d$,
\begin{center}
\begin{tabular}{|c|c|c|}
\hline
value &weight& multiplicity \\[2mm]
\hline
$2^m$&$2^{n-1}-2^{m-1}$&$\frac{2^{m+3d-1}(2^n-1)(2^{n}-2^{n-2d}-2^{n-3d}+2^m-2^{m-d}+1)}{(2^d+1)(2^{2d}-1)}$
\\[2mm]
\hline
$-2^m$&$2^{n-1}+2^{m-1}$&$\frac{2^{m+3d-1}(2^m-1)^2(2^{n}-2^{n-2d}-2^{n-3d}+2^m-2^{m-d}+1)}{(2^d+1)(2^{2d}-1)}$
\\[2mm]
\hline
${2}^{m+d}$&$2^{n-1}-2^{m+d-1}$&$\frac{2^{m-1}(2^{m-d}+1)(2^n-1)(2^m+2^{m-d}+2^{m-2d}+1)}{(2^d+1)^2}$
\\[2mm]
\hline
$-{2}^{m+d}$&$2^{n-1}+2^{m+d-1}$&$\frac{2^{m-1}(2^{m-d}-1)(2^n-1)(2^m+2^{m-d}+2^{m-2d}+1)}{(2^d+1)^2}$
\\[2mm]
\hline
 ${2}^{m+2d}$&$2^{n-1}-2^{m+2d-1}$&
$\frac{2^{m-2d-1}(2^{m-2d}+1)(2^{m-d}-1)(2^{n}-1)}{(2^d+1)(2^{2d}-1)}$
\\[2mm]
\hline
 $-{2}^{m+2d}$&$2^{n-1}+2^{m+2d-1}$&
$\frac{2^{m-2d-1}(2^{m-2d}-1)(2^{m-d}-1)(2^{n}-1)}{(2^d+1)(2^{2d}-1)}$
\\[2mm]
\hline $0$&$2^{n-1}$&$\scriptstyle{(2^{n}-1)(2^{3m-d}-2^{3m-2d}+2^{3m-3d}-2^{3m-4d}+2^{3m-5d}+2^{n-d}-2^{n-2d+1}+2^{n-3d}-2^{n-4d}+1)}$\\[2mm]

\hline $2^m$&$0$&$1$
\\[2mm]
\hline
\end{tabular}
\end{center}
\end{theo}

\begin{proof}
Define
\[\Xi=\left\{(\al,\be,\ga)\in \bF_q^3\left|S(\al,\be,\ga)=0 \right.\right\}\]
and $\xi=\big{|}\Xi\big{|}$.

 Recall $n_i,H_{\al,\be},,r_{\al,\be},A_{\ga}$ in Section 1 and
$N_{i,\veps},n_{i,\veps,}$ in the proof of Lemma \ref{rank}. From
Lemma \ref{det gamma}, if $(\al,\be)\in N_{i,\veps}$, then the
number of $\ga\in \bF_q$ such that $S(\al,\be,\ga)=0$ is
$q_0^s-q_0^{i}$.
 From
Lemma \ref{rank} and Theorem \ref{value dis T} we know that
\begin{itemize}
  \item if $d'=d$ and $(\al,\be)\neq (0,0)$, then
$r_{\al,\be}=s-i$ for some $i\in\{0,1,2\}$. By Lemma \ref{qua} we
have
\begin{equation}\label{value xi1}
{\setlength\arraycolsep{2pt}
\begin{array}{lll}
\xi&=&2^{n}-1+(2^{n}-2^{n-2d})n_2=(2^{3m-d}-2^{n-2d}+1)(2^n-1).\\[2mm]
\end{array}
}
\end{equation}
  \item if $d'=2d$, similarly we have
\begin{equation}\label{value xi2}
{\setlength\arraycolsep{2pt}
\begin{array}{lll}
\xi&=&2^{n}-1+(2^{n}-2^{n-2d})n_{2,1}+(2^{n}-2^{n-4d})n_{4,-1}\\[2mm]
&=&(2^{n}-1)(2^{3m-d}-2^{3m-2d}+2^{3m-3d}-2^{3m-4d}+2^{3m-5d}\\[1mm]
&&\qquad+2^{n-d}-2^{n-2d+1}+2^{n-3d}-2^{n-4d}+1).
\end{array}
}
\end{equation}
\end{itemize}
From (\ref{Wei}) we know that for each non-zero codeword
$c(\al,\be,\ga)=\left(c_0,\cdots,c_{n-1}\right)$ $(n=2^n-1,
c_i=\Tra_{1}^m(\al\pi^{(2^m+1)i})+\Tra_{1}^n(\be \pi^{(2^k+1)i}+\ga
\pi^{i}),\, 0\leq i\leq q-2,\, \text{and}\; (\al,\be,\ga)\in
\bF_{2^m}\times \bF_q^2)$, the Hamming weight of $c(\al,\be,\ga)$ is
\begin{equation}\label{wei c2}
w_H\left(c(\al,\be,\ga)\right)=p^{n-1}-\frac{1}{2}\cdot
S(\al,\be,\ga)
\end{equation}

Combining Lemma \ref{qua} and Theorem \ref{value dis T} we get the
value distribution of $S(\al,\be,\ga)$. As a consequence we have the
weight distribution of $\cC_2$ from (\ref{wei c2}).
\end{proof}

\begin{theo}\label{cor dis}
The collection $\calF$ defined in (\ref{def F}) is a family of
$p$-ary sequences with period $q-1$.
\begin{itemize}
  \item[(i).] If $m/d$ is even, then $\calF$ has family size $2^{3m}+2^m-1$ and correlation distribution as follows.

\begin{center}
\begin{tabular}{|c|c|}
\hline
values & multiplicity \\[2mm]
\hline $\scriptstyle{2^{m}-1}$&$
\begin{array}{ll}
\scriptstyle{(2^{4n+2d-1}-2^{4n+d-1}-2^{4n-1}+2^{7m+2d-1}-2^{7m+d-1}+2^{3n+2d-1}-2^{5m+2d}+2^{5m+d}+2^{5m}}&\\[1mm]
\scriptstyle{\quad -
2^{2n+2d+1}+2^{2n+d+1}+2^{2n}-2^{3m+2d}-2^{3m}+2^{n+2d}-
2^{n+d+1}+2^{m+2d+1}-2^{m+d}-2^{2d}+2^{d})\large{/}(2^{2d}-1)}&
\end{array}
$
\\[2mm]
\hline $\scriptstyle{-2^{m}-1}$&$
\begin{array}{ll}
\scriptstyle{(2^{4n+2d-1}-2^{4n+d-1}-2^{4n-1}-2^{7m+2d-1}+2^{7m+d-1}+2^{7m}+2^{3n+2d-1}-2^{3n}-2^{5m+2d+1}+2^{5m+d}+2^{5m+1}}&\\[1mm]
\scriptstyle{\quad +2^{2n+2d}-2^{2n+1}-2^{3m+d+1}+2^{n+2d}+ 2^{n+d}+
2^n-2^{m}-2^{2d}+2^{d}+2)\large{/}(2^{2d}-1)}&
\end{array}
$
\\[2mm]

\hline
$\scriptstyle{2^{m+d}-1}$&$\scriptstyle{\frac{2^{m-d}(2^{m-d}+1)(2^{m+d}-1)(2^{5m-1}-2^{n}-2^m+1)}{2^{2d}-1}}$\\[2mm]
\hline $\scriptstyle{-2^{m+d}-1}$&$
\begin{array}{ll}
\scriptstyle{(2^{4n-d-1}-2^{7m-1}-2^{7m-2d-1}+2^{3n-d-1}-2^{5m-d}+
2^{2n}-2^{2n-d}+2^{2n-2d}}&\\[1mm]
\scriptstyle{\quad\qquad +2^{3m}+2^{3m-2d}+2^{n+d}-
2^{n+1}-2^{n-d}-2^{n-2d}+3\cdot
2^{m-d}-2^{d+1})\large{/}(2^{2d}-1)}&
\end{array}$\\[2mm]
\hline $\scriptstyle{-1}$&$
\begin{array}{ll}
\scriptstyle{2^{4n-d}-2^{7m-2d}+2^{5m}-2^{5m-d+1}-2^{2n-d+1}+2^{2n-2d+1}}&\\[1mm]
\qquad\quad
\scriptstyle{+2^{3m-d+1}+2^{3m-2d+1}-2^{n+1}-2^{n-2d+1}-2^{m+1}+2^{m-d+1}+2}&
\end{array}
$
\\[2mm]
\hline $\scriptstyle{2^n-1}$ & $\scriptstyle{2^{3m}+2^m-1}$ \\[2mm]
\hline
\end{tabular}
\end{center}

  \item[(ii).] If $k/d$ is even, then $\calF$ has family size $2^{3m}+2^m$ and correlation distribution as follows.

\begin{center}
\begin{tabular}{|c|c|}
\hline
values & multiplicity \\[2mm]
\hline $\scriptstyle{2^{m}-1}$&$
\scriptstyle{\frac{2^{4n+2d-1}-2^{4n+d-1}-2^{4n-1}+2^{7m+2d-1}-2^{7m+d-1}+2^{3n+2d-1}-2^{2n+2d}+2^{2n+d}+2^{2n}-2^{3m+2d}+2^{3m+d}-2^{n+2d}}{2^{2d}-1}}
$
\\[2mm]
\hline $\scriptstyle{-2^{m}-1}$&$
\begin{array}{ll}
\scriptstyle{\left(2^{4n+2d-1}-2^{4n+d-1}-2^{4n-1}-2^{7m+2d-1}+2^{7m+d-1}+2^{7m}+2^{3n+2d-1}-2^{3n}\right.}&\\[1mm]
{\quad\scriptstyle{\left.-2^{5m+2d}+2^{5m}+2^{2n+d}-2^{3m+d}-2^{3m}+
2^n+2^{m+2d}-2^{m}\right)}}\large{/}\scriptstyle{(2^{2d}-1)}&
\end{array}
$
\\[2mm]

\hline
$\scriptstyle{2^{m+d}-1}$&$\scriptstyle{\frac{2^{n-d}(2^{m-d}+1)(2^{m+d}-1)(2^{2n-1}-1)}{2^{2d}-1}}$\\[2mm]
\hline $\scriptstyle{-2^{m+d}-1}$&$
\scriptstyle{\frac{2^{n-d}(2^{m-d}-1)(2^{m+d}-1)(2^{2n-1}-1)}{2^{2d}-1}}$\\[2mm]
\hline $\scriptstyle{-1}$&$
\scriptstyle{2^{4n-d}-2^{7m-2d}+2^{5m}-2^{2n-d+1}+2^{3m-2d+1}-2^{m+1}}
$
\\[2mm]
\hline $\scriptstyle{2^n-1}$ & $\scriptstyle{2^{3m}+2^m}$ \\[2mm]
\hline
\end{tabular}
\end{center}

  \item[(iii).] If $m/d$ and $k/d$ are both odd(that is, $d'=2d$), then $\calF$ has family size $2^{3m}$ and correlation distribution as follows.

\begin{center}
\begin{tabular}{|c|c|}
\hline
values & multiplicity \\[2mm]
\hline
$\scriptstyle{2^m}$&$\scriptstyle{\frac{2^{2n+3d-1}(2^n-2)(2^{n}-2^{n-2d}-2^{n-3d}+2^m-2^{m-d}+1)}{(2^d+1)(2^{2d}-1)}}$
\\[2mm]
\hline
$\scriptstyle{-2^m}$&$\scriptstyle{\frac{2^{3m+3d}(2^{3m-1}-2^n+1)(2^{n}-2^{n-2d}-2^{n-3d}+2^m-2^{m-d}+1)}{(2^d+1)(2^{2d}-1)}}$
\\[2mm]
\hline
$\scriptstyle{{2}^{m+d}}$&$\scriptstyle{\frac{2^{3m}(2^{2n-d-1}+2^{3m-1}-2^{n-d}-2^m+2^d)(2^m+2^{m-d}+2^{m-2d}+1)}{(2^d+1)^2}}$
\\[2mm]
\hline
$\scriptstyle{-{2}^{m+d}}$&$\scriptstyle{\frac{2^{2n-1}(2^{m-d}-1)(2^n-2)(2^m+2^{m-d}+2^{m-2d}+1)}{(2^d+1)^2}}$
\\[2mm]
\hline
 $\scriptstyle{{2}^{m+2d}}$&
$\scriptstyle{\frac{2^{2n-2d-1}(2^{m-2d}+1)(2^{m-d}-1)(2^{n}-2)}{(2^d+1)(2^{2d}-1)}}$
\\[2mm]
\hline
 $\scriptstyle{-{2}^{m+2d}}$&
$\scriptstyle{\frac{2^{3m}(2^{m-d}-1)(2^{2n-2d-1}-2^{3m-2d-1}-2^{n-2d}+2^{m-2d}+1)}{(2^d+1)(2^{2d}-1)}}$
\\[2mm]
\hline $\scriptstyle{0}$&$\scriptstyle{2^{3m}(2^{n}-2)(2^{3m-d}-2^{3m-2d}+2^{3m-3d}-2^{3m-4d}+2^{3m-5d}+2^{n-d}-2^{n-2d+1}+2^{n-3d}-2^{n-4d}+1)}$\\[2mm]

\hline $\scriptstyle{2^m}$&$\scriptstyle{2^{3m}}$
\\[2mm]
\hline
\end{tabular}
\end{center}

\end{itemize}
\end{theo}
\begin{proof}

For any possible value $\ka$ and $1\leq i, j\leq 3$, define
$M_{\ka}(\calF_i, \calF_j)$ to be the frequency of $\ka$ in
correlation values
 between two sequences in $\calF_i$ and $\calF_j$ by any shift, respectively.  Then the correlation distribution of sequences in
 $\calF$
could be obtained if we can calculate all of the $M_{\ka}(\calF_i,
\calF_j)$ . We will deal with it case by case.
\begin{itemize}
  \item The correlation function between
$a_{\al_1,\be_1}$ and $a_{\al_2,\be_2}$ by a shift $\tau$ ($0\leq
\tau\leq q-2$) is
\[
\begin{array}{ll}
&C_{(\al_1,\be_1),(\al_2,\be_2)}(\tau)=\sum\limits_{\lambda=0}^{q-2}(-1)^{a_{\al_1,\be_1}({\lambda})-
a_{\al_2,\be_2}({\lambda+\tau})}\\[2mm]
&\qquad =\sum\limits_{\lambda=0}^{q-2}(-1)^{\Tra_1^m(\al_1
\pi^{\lambda(2^m+1)})+\Tra_1^n(\be_1
  \pi^{\lambda(2^k+1)}+\pi^{\lambda})-\Tra_1^m(\al_2 \pi^{(\lambda+\tau)(2^m+1)})-\Tra_1^n(\be
  \pi^{(\lambda+\tau)(2^k+1)}+\pi^{\lambda+\tau})}\\[2mm]
  &\qquad = S(\al',\be',\ga')-1
  \end{array}
\]
 where
 \begin{equation}\label{coe cor}
 \al'=\al_1-\al_2 \pi^{\tau(2^m+1)},\quad
 \be'=\be_1-\be_2\pi^{\tau(2^k+1)},\quad \ga'=1-\pi^{\tau}.
 \end{equation}

 Fix $(\al_2,\be_2)\in
\bF_{2^m}\times \bF_q$, when $(\al_1,\be_1)$ runs through
$\bF_{2^m}\times \bF_q$ and $\tau$ takes values from $0$ to $q-2$,
$(\al',\be',\ga')$ runs through $\bF_{2^m}\times
\bF_q\times\left\{\bF_{q}\big{\backslash}\{1\}\right\}$ exactly one
time.

For any possible value $\kappa$ of $S(\al,\be,\ga)$, define

\begin{equation}\label{def sk}
s_{\kappa}=\#\left\{(\al,\be,\ga)\in \bF_{2^m}\times \bF_q\times
\bF_q\,\displaystyle{|}\,S(\al,\be,\ga)=\kappa+1\right\}
\end{equation}

\begin{equation}\label{def sk1}s^1_{\kappa}=\#\left\{(\al,\be)\in \bF_{2^m}\times
\bF_q\,\big{|}\,S(\al,\be,1)=\kappa+1\right\}
\end{equation}
 and
\begin{equation}\label{def tk}t_{\kappa}=\#\left\{(\al,\be)\in \bF_{2^m}\times
\bF_q\,\displaystyle{|}\,T(\al,\be)=\kappa+1\right\}.
\end{equation}

By Lemma \ref{q-2} we have
\begin{equation}\label{rel sk1 sk tk}
s_{\kappa}^1=\frac{1}{2^n-1}\times (s_{\kappa}-t_{\kappa}).
\end{equation}

Hence we get
\[M_{\kappa}(\calF_1,\calF_1)=2^{3m}\cdot \left(s_{\kappa}-s_{\kappa}^1\right)=2^{3m}\cdot\left(\frac{2^n-2}{2^n-1}\cdot s_{\kappa}+\frac{1}{2^n-1}\cdot t_{\kappa}\right).\]

  \item For the case $m/d$ or $k/d$ is even.  The cross correlation function between
$a_{\al_1,\be_1}$ and $a_{\be_2}$ by a shift $\tau$ ($0\leq
\tau\leq q-2$) is
\[
\begin{array}{ll}
&C_{(\al_1,\be_1),\be_2}(\tau)=\sum\limits_{\lambda=0}^{q-2}(-1)^{a_{\al_1,\be_1}({\lambda})-
a_{\be_2}({\lambda+\tau})}\\[2mm]
&\qquad =\sum\limits_{\lambda=0}^{q-2}(-1)^{\Tra_1^m(\al_1
\pi^{\lambda(2^m+1)})+\Tra_1^n(\be_1
  \pi^{\lambda(2^k+1)}+\pi^{\lambda})-\Tra_1^m(\pi^{(\lambda+\tau)(2^m+1)})-\Tra_1^n(\be_2
  \pi^{(\lambda+\tau)(2^k+1)})}\\[2mm]
  &\qquad = S(\al',\be',1)-1
  \end{array}
\]
 where
 $
 \al'=\al_1- \pi^{\tau(2^m+1)},
 \be'=\be_1-\be_2\pi^{\tau(2^k+1)}.
$

 Fix $0\leq \tau\leq q-2$ and $\be_2=\pi^{i}$ for some $0\leq i\leq 2^m-2$, when $(\al_1,\be_1)$ runs through
$\bF_{2^m}\times \bF_q$, $(\al',\be')$ runs through $\bF_{2^m}\times
\bF_q$ exactly one time.  By (\ref{def sk}) and (\ref{def sk1}) we
get
\[M_{\kappa}(\calF_1,\calF_2)=M_{\kappa}(\calF_2,\calF_1)=(q-1)(2^{m}-1)\cdot s_{\kappa}^1=(2^m-1)(s_{\ka}-t_{\ka}).\]

  \item For the case $k/d$ is even.  The cross correlation function between
$a_{\al_1,\be_1}$ and $a$ by a shift $\tau$ ($0\leq \tau\leq q-2$)
is
\[
\begin{array}{ll}
&C_{(\al_1,\be_1)}(\tau)=\sum\limits_{\lambda=0}^{q-2}(-1)^{a_{\al_1,\be_1}({\lambda})-
a({\lambda+\tau})}\\[2mm]
&\qquad =\sum\limits_{\lambda=0}^{q-2}(-1)^{\Tra_1^m(\al_1
\pi^{\lambda(2^m+1)})+\Tra_1^n(\be_1
  \pi^{\lambda(2^k+1)}+\pi^{\lambda})-\Tra_1^n(
  \pi^{(\lambda+\tau)(2^k+1)})}\\[2mm]
  &\qquad = S(\al_1,\be',1)-1
  \end{array}
\]
 where
 $
 \be'=\be_1-\pi^{\tau(2^k+1)}.
 $

For fixed $\tau$, $0\leq \tau\leq q-2$, when $\be_1$ runs through
$\bF_q$, $\be'$ runs through $\bF_q$ exactly one time.
 By (\ref{def sk}) and (\ref{def sk1}) we get
\[M_{\kappa}(\calF_1,\calF_3)=M_{\kappa}(\calF_3,\calF_1)=(q-1)\cdot s_{\kappa}^1={s_{\ka}-t_{\ka}}.\]

  \item For the case $m/d$ or $k/d$ is even.  The cross correlation function between
$a_{\be_1}$ and $a_{\be_2}$ by a shift $\tau$ ($0\leq
\tau\leq q-2$) is
\[
\begin{array}{ll}
&C_{\be_1,\be_2}(\tau)=\sum\limits_{\lambda=0}^{q-2}(-1)^{a_{\be_1}({\lambda})-
a_{\be_2}({\lambda+\tau})}\\[2mm]
&\qquad =\sum\limits_{\lambda=0}^{q-2}(-1)^{\Tra_1^m(
\pi^{\lambda(2^m+1)})+\Tra_1^n(\be_1
  \pi^{\lambda(2^k+1)})-\Tra_1^m(\pi^{(\lambda+\tau)(2^m+1)})-\Tra_1^n(\be_2
  \pi^{(\lambda+\tau)(2^k+1)})}\\[2mm]
  &\qquad = T(\al',\be')-1
  \end{array}
\]
 where
 $
 \al'=1- \pi^{\tau(2^m+1)},
 \be'=\be_1-\be_2\pi^{\tau(2^k+1)}.
$

When $(\be_1,\be_2)$ runs through $\bF_{q}\times \bF_q$ and $\tau$
takes value from $0$ to $q-2$, $(\al',\be')$ runs through
$\bF_{2^m}\large{\backslash}\{1\}\times \bF_q$ exactly $(2^m+1)\cdot
q$ times.

Fix $\be_2=0$. When $\be_1$ runs through $\bF_{q}$ and $\tau$ takes
value from $0$ to $q-2$, $(\al',\be')$ runs through
$\bF_{2^m}\large{\backslash}\{1\}\times \bF_q$ exactly $(2^m+1)$
times. By symmetry,  this statement is also valid if we exchange
$\be_1$ and $\be_2$ to each other.

When $\be_1=\be_2=0$ and $\tau$ takes value from $0$ to $q-2$, then
$\be'=0$ and $\al'$ runs through $\bF_{2^m}\large{\backslash}\{1\}$
exactly $(2^m+1)$ times. In this case $\phi_{\al',0}(x)$ defined in
(\ref{def phi}) is $\al' x^{2^m}$. Hence $T(\al',0)=2^n$ or $-2^{m}$
provided that $\al'=0$
 or not. Define
 \[
l_{\ka}=\left\{
\begin{array}{ll}
1,&\text{if}\;\ka=2^n-1\\[1mm]
2^m-2, &\text{if}\;\ka=-2^{m}-1\\[1mm]
0,&\text{otherwise}
\end{array}
\right.
 \]

 Define
\begin{equation}\label{def tk0}
\displaystyle{t_{\ka}^0=\#\left\{\be\in
\bF_q\large{|}\;T(0,\be)=\ka+1\right\}}.
\end{equation}
Then a routine calculation shows that

\[
t_{\ka}^0=\left\{\begin{array}{cl}1, & \ka=2^n-1\\[2mm]
2^n-1, &\ka=-1\;\text{and}\;k/d\;\text{is even}\\[2mm]
\frac{2^d(2^n-1)}{2^d+1},&\ka=2^m-1\;\text{and}\;m/d\;\text{is even}\\[2mm]
\frac{2^n-1}{2^d+1},&\ka=-2^{m+d}-1\;\text{and}\;m/d\;\text{is even}\\[2mm]0,&\text{otherwise}.\end{array}\right.
\]

By Inclusion-Exclusion principle, Lemma \ref{q-2},  (\ref{def tk})
and (\ref{def tk0}) we get
\[
\begin{array}{rcl}
M_{\kappa}(\calF_2,\calF_2)&=&\left(\frac{1}{2^m+1}\right)^2\cdot
(2^m+1)\left[(2^n-2)\left(\frac{2^m-2}{2^m-1}(t_{\ka}-t_{\ka}^0)+t_{\ka}^0\right)+l_{\ka}\right]\\[2mm]
&=&\frac{(2^m-2)(2^n-2)}{2^n-1}t_{\ka}+\frac{2^n-2}{2^n-1}t_{\ka}^0+\frac{1}{2^m+1}l_{\ka}.
\end{array}
\]

  \item For the case $k/d$ is even.  The cross correlation function between
$a_{\be}\in \calF_2$ and $a\in \calF_3$ by a shift $\tau$ ($0\leq
\tau\leq q-2$) is
\[
\begin{array}{ll}
&C_{\be}(\tau)=\sum\limits_{\lambda=0}^{q-2}(-1)^{a_{\be}({\lambda})-
a({\lambda+\tau})}\\[2mm]
&\qquad =\sum\limits_{\lambda=0}^{q-2}(-1)^{\Tra_1^m(
\pi^{\lambda(2^m+1)})+\Tra_1^n(\be
  \pi^{\lambda(2^k+1)})-\Tra_1^n(
  \pi^{(\lambda+\tau)(2^k+1)})}\\[2mm]
  &\qquad = T(1,\be')-1
  \end{array}
\]
 where
 $
 \be'=\be-\pi^{\tau(2^k+1)}.
$

When $\be$ runs through
$\bF_{q}^*$ and $\tau$ takes value from $0$ to $q-2$,
$\be'$ runs through $\bF_q$ exactly $q-2$ times except for $0$, on which $\be'$ has frequency $q-1$.

Define
\[t_{\ka}^1=\#\{\be\in \bF_q\;|\;T(1,\be)=\ka+1\}.\]
 Then from Lemma (\ref{q-2}) we have
\[t_{\ka}^1=\frac{t_{\ka}-t_{\ka}^0}{2^m-1}.\]

 Note that $T(1,0)=-2^m$. By  Lemma \ref{q-2} we get
\[\begin{array}{rcl}M_{\kappa}(\calF_2,\calF_3)&=&M_{\kappa}(\calF_3,\calF_2)=\frac{1}{2^m+1}\left((q-2)t_{\ka}^1+\delta(\ka,0)\right)\\[2mm]
&=&\frac{2^n-2}{2^n-1}\left(t_{\ka}-t_{\ka}^0\right)+\frac{1}{2^m+1}\delta(\ka,-2^m-1)
\end{array}
\] where the
Hermitian symbol $\delta(\ka,-2^m-1)=1$ if $\ka=-2^m-1$, and zero
otherwise.

  \item For the case $k/d$ is even.  The auto-correlation function of
$a\in\calF_3$ by a shift $\tau$ ($0\leq
\tau\leq q-2$) is
\[
\begin{array}{ll}
&C(\tau)=\sum\limits_{\lambda=0}^{q-2}(-1)^{a({\lambda})-
a({\lambda+\tau})}\\[2mm]
&\qquad =\sum\limits_{\lambda=0}^{q-2}(-1)^{\Tra_1^n(
  \pi^{\lambda(2^k+1)})-\Tra_1^n(
  \pi^{(\lambda+\tau)(2^k+1)})}\\[2mm]
  &\qquad = T(0,1-\pi^{\tau(2^k+1)})-1.
  \end{array}
\]
Since $\gcd(2^k+1,2^n-1)=1$, $C(\tau)=-1$ for $\tau\neq 0$ and
$C({0})=q-1$. Then
\[M_{\ka}(\calF_3,\calF_3)=\left\{
\begin{array}{ll}
q-2, &\text{if}\;\ka=-1\\[2mm]
1,&\text{if}\;\ka=q-1\\[2mm]
0,&\text{otherwise}.
\end{array}
\right.
\]
In total, sum up all the $M_{\ka}(\calF_i,\calF_i)$ for $1\leq
i,j\leq 3$ and the result follows from Theorem \ref{value dis T} and
Theorem \ref{value dis S}.
\end{itemize}
\end{proof}
\begin{remark}
The case $d'=d=1$ has been shown in \cite{Zen Liu}, Theorem 23 and
25.
\end{remark}

\section{Conclusion}

\quad In this paper we have studied the exponential sums
$\sum\limits_{x\in \bF_q}(-1)^{\Tra_1^m (\alpha
x^{2^{m}+1})+\Tra_1^n(\beta x^{2^k+1})}$ and $\sum\limits_{x\in
\bF_q}(-1)^{\Tra_1^m (\alpha x^{2^{m}+1})+\Tra_1^n(\beta
x^{2^k+1}+\ga x)}$ with $\al\in \bF_{2^m},(\be,\ga)\in \bF_q^2$.
After giving the value distribution of $\sum\limits_{x\in
\bF_q}\zeta_p^{\Tra_1^m (\alpha x^{2^{m}+1})+\Tra_1^n(\beta
x^{2^k+1})}$ and $\sum\limits_{x\in \bF_q}(-1)^{\Tra_1^m (\alpha
x^{2^{m}+1})+\Tra_1^n(\beta x^{2^k+1}+\ga x)}$, we determine the
correlation distribution among a family of sequences, and the weight
distributions of the cyclic codes $\cC_1$ and $\cC_2$. These results
generalize  \cite{Kasa1},\cite{Vand2} and \cite{Zen Liu}.

\section{Acknowledgements}
\quad The authors will thank the anonymous referees for their
helpful comments.

\end{document}